\documentclass[lettersize,journal]{IEEEtran}

\usepackage{amsmath,amssymb,amsfonts}
\usepackage{amsthm}
\usepackage{bm}

\usepackage{algorithm}
\usepackage{algorithmic}

\usepackage{graphicx}
\usepackage[caption=false,font=normalsize,labelfont=sf,textfont=sf]{subfig}

\usepackage{array}
\usepackage{booktabs}
\usepackage{makecell}
\usepackage{diagbox}
\usepackage{multirow}
\usepackage{tabularx}
\usepackage[table]{xcolor}

\usepackage{textcomp}
\usepackage{stfloats}
\usepackage{url}
\usepackage{verbatim}
\usepackage{cite}
\usepackage{float}
\usepackage{balance}

\usepackage[bookmarks=false,hidelinks]{hyperref}

\def\BibTeX{{\rm B\kern-.05em{\sc i\kern-.025em b}\kern-.08em
    T\kern-.1667em\lower.7ex\hbox{E}\kern-.125emX}}

\newtheoremstyle{mynonitalic}%
  {3pt}
  {3pt}
  {\normalfont}
  {}
  {\bfseries}
  {.}
  {.5em}
  {}

\theoremstyle{mynonitalic}
\newtheorem{theorem}{Theorem}
\newtheorem{proposition}{Proposition}
\newtheorem{lemma}{Lemma}

\newcolumntype{C}{>{\centering\arraybackslash}X}

\makeatletter
\def\fs@myruled{%
  \fs@ruled
  \def\@fs@pre{\vspace*{0.12in}\hrule height.8pt depth0pt \kern2pt}%
  \def\@fs@mid{\kern2pt\hrule\kern2pt}%
  \def\@fs@post{\kern2pt\hrule\relax}%
}
\makeatother

\floatstyle{myruled}
\restylefloat{algorithm}

\begin{document}
\title{Parallel Successive Cancellation Perturbation-Enhanced Decoding of Polar Codes via Offline Variance Design}
\author{Changwei Tu, Xuanyu Li and Kai Niu
\thanks{This work is supported by the National Natural Science Foundation of China under Grant 62321001 and Grant 62471054.
  \emph{(Corresponding author: Kai Niu.)}
  
The authors are with the Key Laboratory of Universal Wireless Communications, Ministry of Education,
Beijing University of Posts and Telecommunications, Beijing, China.
Email: \{tuchangwei, lixuanyu, niukai\}@bupt.edu.cn.
}

  }

\markboth{Journal of \LaTeX\ Class Files,~Vol.~18, No.~9, September~2020}%
{How to Use the IEEEtran \LaTeX \ Templates}

\maketitle

\begin{abstract}
Successive cancellation perturbation-enhanced (SCP) decoding improves the performance of finite-length polar codes by performing multiple SC decoding attempts with receiver-side perturbations. However, many existing perturbation schemes generate or update subsequent perturbations according to the outcomes of previous decoding attempts, resulting in additional decoding latency.
In this paper, we propose an offline variance design (OVD) method for parallel SCP (PSCP) decoding of short- and medium-length polar codes. 
First, we formulate the exact recovery objective conditioned on ordinary SC failure and classify failed frames by the position of
the first genie-aided intrinsic error and the number of subsequent intrinsic errors.
We also derive a consistent Gaussian representation of the perturbed channel that preserves min-sum SC hard decisions.
Second, we construct a class-based approximation of the recovery objective using Gaussian approximation and backward recursions,
accounting for both error correction and new errors introduced by perturbations.
We prove that both the exact and analytical objectives are nondecreasing and exhibit diminishing marginal gains as independent branches are added.
Third, we develop a greedy algorithm to select variances from a finite candidate set for a given code, signal-to-noise ratio (SNR), and number of perturbation branches.
All variances are determined offline, allowing the original SC branch and all perturbation branches to start simultaneously.
Simulations for rate-$1/2$ polar codes of lengths $64$, $128$, $256$, and $512$ show that OVD-PSCP achieves lower block error rates (BLERs) than conventional SCP with the same number of perturbation branches. The gains are larger for shorter codes and increase as the number of perturbation branches grows.
\end{abstract}


\section{Introduction}
Polar codes are the first class of error-correction codes with explicit construction that provably achieve the channel capacity, and have low encoding/decoding complexities \cite{b1}. Therefore, polar codes were selected as the fifth-generation (5G) mobile communication standard \cite{b2}.

In recent years, with the ongoing development of 6G standards, channel coding has been facing increasingly stringent requirements on ultra-reliable and low-latency communications (URLLC). To address the latency issue, \cite{b3} optimizes Arıkan’s successive cancellation (SC) decoding by exploiting special nodes.
By trading off a certain degree of subchannel reliability, the proposed approach can achieve a decoding throughput at the terabit-per-second (Tb/s) level.
However, for polar codes with short-to-moderate code lengths, polarization is insufficient, which limits the decoding performance of SC decoding.
To mitigate this limitation, 
researchers proposed successive cancellation list (SCL) decoding \cite{b4}, and further improved the decoding performance by concatenating a cyclic redundancy check (CRC) \cite{b5}. Nevertheless, these methods suffer from high decoding complexity.

To achieve a better trade-off between decoding performance and complexity,  \cite{b6} proposes adding artificial noise to the received  vector 
when SC decoding fails to facilitate decoding. This method is referred to SC perturbation-enhanced (SCP) decoding. 
Compared with other performance-enhancement decoding algorithms, SCP supports parallel decoding attempts at the receiver, thereby satisfying low-latency requirements.

Building upon this idea, \cite{b7} proposes a dynamic perturbation-enhanced scheme that increases the artificial-noise variance whenever duplicate decoded codewords occur at the receiver, thereby generating more candidate codewords. 
After that,  \cite{b8} demonstrates that, in the asymptotic regime, the first perturbation improves decoding performance with probability 1/2,
and further establishes an equivalence between channel-side perturbations and source-side perturbations. Moreover, \cite{b9} combines perturbation decoding with bit-flipping decoding to further improve decoding performance. For longer polar codes, \cite{b10} proposes a dedicated hybrid perturbation decoding scheme, which achieves decoding performance close to CRC-aided SCL (CA-SCL) decoding. More recently, \cite{b11} proposes an improved  perturbation-enhanced  decoder that 
perturbs only the a posteriori LLRs at selected unreliable positions,  reducing perturbation complexity while maintaining similar decoding performance.

Despite these advances, many subsequent improvements first identify a failed frame and then generate or update the perturbation attempts in 
sequence \cite{b7,b9,b10}. This serial control increases decoding  latency. Moreover, existing methods often use an empirically selected  variance or an online variance-update rule. 
The offline design of a multiset of perturbation variances for a fixed number of parallel branches has received limited attention.

To address these limitations, we propose an offline variance design (OVD) method for parallel successive cancellation perturbation-enhanced
(PSCP) decoding of short- and medium-length polar codes.
The resulting decoder is referred to as OVD-PSCP.
For a given code, operating signal-to-noise ratio (SNR), and number $T$ of perturbation branches, OVD selects $T$ perturbation variances
from a finite candidate set.
The selection is based on an analytical approximation of the recovery probability conditioned on ordinary SC decoding failure. 
Since all variances are determined offline, the original SC branch and all perturbation branches can start simultaneously.

The main contributions of this paper are summarized as follows.

\begin{itemize}

\item
We formulate the exact frame-level recovery objective for PSCP decoding.
Since this objective is difficult to evaluate directly, we classify SC-failed frames by the position of the first genie-aided intrinsic
error and the number of subsequent intrinsic errors.
To characterize the effect of perturbation, we derive a consistent Gaussian representation that is hard-decision equivalent to the
perturbed channel under min-sum SC decoding.

\item
We develop an analytical recovery model for PSCP decoding. Based on the decision-equivalent Gaussian representation, we use
Gaussian approximation, local repair and damage probabilities, and backward recursions to estimate the class priors and single-branch recovery probabilities.
These probabilities are then combined to construct a tractable analytical approximation of the multi-branch recovery objective.
We further prove that both the exact and analytical objectives are nondecreasing and exhibit diminishing marginal gains as independent
perturbation branches are added.

\item
Based on the analytical objective, we propose a greedy offline variance design algorithm for PSCP decoding with a fixed number of perturbation branches.
We first construct a finite candidate variance set from the local perturbation model.
The algorithm then selects the candidate variance with the largest marginal gain at each step, allowing the same variance to be selected more than once.
The selected variances are then used for online PSCP decoding, in which the original SC branch and all perturbation branches can start simultaneously.
We also analyze the computational complexity and memory requirements of the offline design.

\end{itemize}

The remainder of this paper is organized as follows.
Section~II introduces polar codes, SC decoding, and PSCP decoding.
Section~III develops the analytical recovery model and its structural properties.
Section~IV presents the candidate variance set, the OVD-PSCP algorithm, and the offline design complexity.
Section~V reports the BLER performance and candidate codeword statistics.
Finally, Section~VI concludes the paper.

\section{Preliminaries}

\subsection{Polar Codes}
Let $\mathcal{P}(N,K+r)$ denote a polar code with length $N=2^n$ and $K+r$ non-frozen bits, where $K$ is the number of information bits and $r$ is the CRC length. Let $M=K+r$ denote the number of non-frozen bits. After channel polarization, the $M$ most reliable subchannels are selected as non-frozen positions, while the remaining $N-M$ subchannels are assigned to frozen bits. Let $\mathcal{I}$ and $\mathcal{F}$ denote the sets of non-frozen and frozen bit positions, respectively, where $|\mathcal{I}|=M$, $|\mathcal{F}|=N-M$, and $\mathcal{F}=\{0,1,\ldots,N-1\}\setminus\mathcal{I}$.

For convenience, the elements of $\mathcal{I}$ are arranged in ascending order as $\mathcal{I}=\{a_0,a_1,\ldots,a_{M-1}\}$, where $0\leq a_0<a_1<\cdots<a_{M-1}\leq N-1$. Let $\mathbf{u}=[u_0,u_1,\ldots,u_{N-1}]$ and $\mathbf{c}=[c_0,c_1,\ldots,c_{N-1}]$ denote the source-bit vector and the codeword vector, respectively. Then, the polar encoding process can be expressed as $\mathbf{c}=\mathbf{u}\mathbf{G}_N$, where $\mathbf{G}_N=\begin{bmatrix}1&0\\1&1\end{bmatrix}^{\otimes n}$. Since the CRC bits do not carry independent information, the effective rate of the CRC-aided polar code is defined as $R=K/N$.

The codeword bit $c_i\in\{0,1\}$ is modulated by binary phase-shift keying (BPSK) as $x_i=1-2c_i$. Then, the modulated symbol passes through an additive white Gaussian noise (AWGN) channel. The received symbol is given by $y_i=x_i+n_i$, where $\{n_i\}$ are i.i.d.\ Gaussian noise samples with $n_i\sim\mathcal{N}(0,\sigma^2)$.

\subsection{SC Decoding of Polar Codes}

Let $L(y_i)=\frac{2y_i}{\sigma^2}$ denote the channel LLR at index $i$, and let $L(u_i)$ denote the source-side LLR at index $i$.
For a length-$2$ polar code, SC decoding proceeds as follows:
The $f$-operation computes the decision LLR of $u_0$ as
\begin{equation}
\begin{split}
L(u_0)&=f\big(L(y_0),L(y_1)\big)\\
&=\mathrm{sgn}\!\big(L(y_0)\big)\,\mathrm{sgn}\!\big(L(y_1)\big)\
\!\min\!\big(|L(y_0)|,|L(y_1)|\big),
\end{split}
\end{equation}
and a hard decision on $L(u_0)$ yields the estimate $\hat{u}_0$. Based on the value of $\hat{u}_0$, the $g$-operation updates the decision LLR of $u_1$ as
\begin{equation}
L(u_1)=g(L(y_0),L(y_1),\hat{u}_0)=(-1)^{\hat{u}_0}L(y_0)+L(y_1).
\end{equation}
Finally, a hard decision on $L(u_1)$ yields the estimate $\hat{u}_1$.
The hard decision in the above steps is given by
\begin{equation}
\hat{u}_i=
\begin{cases}
1, & L(u_i)<0,\ \ i\in\mathcal{I},\\
0, & \text{otherwise}.
\end{cases}
\end{equation}
By recursively applying the $f$- and $g$-operations over the decoding tree, SC decoding produces $\hat{\mathbf{u}}=[\hat{u}_0,\hat{u}_1,\ldots,\hat{u}_{N-1}]$ for a length-$N$ polar code.

\subsection{Parallel SC Perturbation-Enhanced Decoding}

Based on the receiver-side SC perturbation-enhanced decoding scheme in \cite{b6}, we describe a PSCP decoder.
It contains one original SC branch and $T$ perturbation branches, and all $T+1$ branches are executed in parallel.

Let $V=\{v_0,v_1,\ldots,v_{T-1}\}$ denote the multiset of perturbation variances, where $v_t$ is the variance used by the $t$-th perturbation branch.
For $t=0,1,\ldots,T-1$, a Gaussian perturbation vector is added to the channel LLR vector, yielding
\begin{equation}
\mathbf{L}^{(t)}(p)
=
\mathbf{L}(y)+\mathbf{n}_p^{(t)}.
\label{eq:parallel_perturbed_llr}
\end{equation}
The $i$-th element of $\mathbf{n}_p^{(t)}$ satisfies $n_{p_i}^{(t)}\sim\mathcal{N}(0,v_t)$ for $i=0,1,\ldots,N-1$.
For each perturbation branch, the perturbation samples are i.i.d.\ across the LLR positions, and the perturbation vectors are
independent across different branches. Each perturbation branch applies SC decoding to its perturbed LLR vector, while the original SC branch directly uses the
channel LLR vector $\mathbf{L}(y)$ without perturbation.

The CRC check is applied to each decoded candidate. If at least one candidate passes the CRC check, a valid estimate is returned; otherwise, a decoding failure is declared.

\section{Performance Analysis of SC Perturbation-Enhanced Decoding}
\label{sec:recovery_model}

This section develops an analytical recovery model for offline variance design in PSCP decoding.
We first formulate the exact frame-level recovery objective conditioned on ordinary SC decoding failure.
Then, we divide SC-failed frames into genie-aided error classes and estimate the class priors and single-branch recovery probabilities.
These estimates are combined to construct a tractable analytical approximation of the multi-branch recovery objective.
Finally, we prove that both the exact and analytical objectives are nondecreasing and exhibit diminishing marginal gains as independent perturbation branches are added.

\subsection{Frame-Level Recovery Objective}

Let $\hat{\mathbf{u}}^{\mathrm{SC}}$ denote the source-vector estimate
produced by ordinary SC decoding, and define the SC failure event as
\begin{equation}
 \mathrm{Fail}_{\mathrm{SC}}
 \triangleq
 \left\{
 \hat{\mathbf{u}}^{\mathrm{SC}}
 \neq
 \mathbf{u}
 \right\}.
\label{eq:ovd_sc_failure}
\end{equation}
For a fixed failed frame $\omega\in\mathrm{Fail}_{\mathrm{SC}}$, let
$\hat{\mathbf{u}}(v)$ denote the output of one perturbation branch with
perturbation variance $v$. Its single-branch recovery probability is defined
as
\begin{equation}
 s_{\omega}(v)
 \triangleq
 \Pr\!\left(
 \hat{\mathbf{u}}(v)=\mathbf{u}
 \mid \omega
 \right).
\label{eq:ovd_single_branch}
\end{equation}

For the perturbation-variance multiset $V$ defined in Section~II, the perturbation vectors
are generated independently across branches. Therefore, the conditional
probability that at least one of the $T$ branches recovers the transmitted
source vector is
\begin{equation}
 P_{\omega}(V)
 \triangleq
 1-
 \prod_{v_t\in V}
 \left[1-s_{\omega}(v_t)\right].
\label{eq:ovd_fixed_frame_recovery}
\end{equation}
This probability depends on the selected variance values and their
multiplicities, but not on the order of the branches.

Averaging over ordinary-SC failed frames gives the exact frame-level
objective
\begin{equation}
 J(V)
 \triangleq
 \mathbb{E}_{\omega}\!\left[
 P_{\omega}(V)
 \mid
 \mathrm{Fail}_{\mathrm{SC}}
 \right].
\label{eq:ovd_true_objective}
\end{equation}
Thus, $J(V)$ is the probability that the transmitted source vector is
recovered by at least one perturbation branch, conditioned on an ordinary-SC
decoding failure.

\subsection{Genie-Aided Error Classification}
\label{subsec:error_classification}

We next partition the failed frames into error classes. For the $\ell$-th
non-frozen bit $u_{a_\ell}$, let $\hat{u}_{a_\ell}^{\mathrm{G}}$ denote the
estimate produced by the genie-aided SC decoder, which sets all preceding
non-frozen decisions to their transmitted values. We define
\begin{equation}
 X_\ell
 \triangleq
 \mathbf{1}\!\left\{
 \hat{u}_{a_\ell}^{\mathrm{G}}
 \neq
 u_{a_\ell}
 \right\},
 \qquad
 H
 \triangleq
 \sum_{\ell=0}^{M-1} X_\ell .
\label{eq:ovd_genie_indicator}
\end{equation}
Here, $X_\ell$ indicates whether an intrinsic error occurs at the $\ell$-th
non-frozen decision, and $H$ counts the total number of such errors. Since
the ordinary and genie-aided SC decoders use identical preceding decisions
before the first ordinary-SC error, an ordinary-SC decoding failure occurs
if and only if $H>0$.

For $0\leq\ell<M$, suppose that the first intrinsic error occurs at
$a_\ell$, and let $m$ denote the number of subsequent intrinsic errors.
Then, $0\leq m<M-\ell$ and $H=m+1$. Define
\begin{equation}
 \mathcal{E}_{\ell,m}
 \triangleq
 \left\{
 \begin{array}{l}
 X_s=0,\quad 0\leq s<\ell,\\[1mm]
 X_\ell=1,\\[1mm]
 \displaystyle\sum_{s=\ell+1}^{M-1}X_s=m
 \end{array}
 \right\}.
\label{eq:ovd_true_class}
\end{equation}
The valid class-index set is
\begin{equation}
 \mathcal{S}
 \triangleq
 \left\{
 (\ell,m):
 0\leq\ell<M,\;
 0\leq m<M-\ell
 \right\}.
\label{eq:ovd_class_set}
\end{equation}
The events
$\{\mathcal{E}_{\ell,m}:(\ell,m)\in\mathcal{S}\}$ are mutually exclusive
and form a partition of $\mathrm{Fail}_{\mathrm{SC}}$.

For each class, define
\begin{align}
 \pi_{\ell,m}
 &\triangleq
 \Pr\!\left(
 \mathcal{E}_{\ell,m}
 \mid
 \mathrm{Fail}_{\mathrm{SC}}
 \right),
\label{eq:ovd_true_class_prior}\\
 J_{\ell,m}(V)
 &\triangleq
 \mathbb{E}_{\omega}\!\left[
 P_{\omega}(V)
 \mid
 \mathcal{E}_{\ell,m}
 \right].
\label{eq:ovd_true_class_coverage}
\end{align}
Since the error classes form a partition of the SC failure event, the frame-level objective can be written as
\begin{equation}
 J(V)
 =
 \sum_{(\ell,m)\in\mathcal{S}}
 \pi_{\ell,m}J_{\ell,m}(V).
\label{eq:ovd_true_decomposition}
\end{equation}
However, the class priors and class-wise recovery probabilities are generally not available in closed form. Therefore, we develop an analytical model below.

\subsection{Offline Analytical Recovery Model}
\label{subsec:analytical_recovery}

\subsubsection{Gaussian Approximation and Decision-Equivalent Channel}

We use the Gaussian approximation (GA) in \cite{b12} to model the
marginal distributions of genie-aided decision LLRs. By symmetry, the
all-zero codeword is assumed to be transmitted.

Assume $\mu_{\mathrm{ch}}>0$. The unperturbed channel LLR follows
\begin{equation}
 L_{\mathrm{ch}}
 \sim \mathcal{N}(\mu_{\mathrm{ch}},2\mu_{\mathrm{ch}}).
\label{eq:ovd_consistent_channel}
\end{equation}
Let $\mathcal{G}_\ell(\cdot)$ denote the GA mapping from the channel-LLR
mean to the decision-LLR mean for the non-frozen bit $u_{a_\ell}$. Define
\begin{equation}
 \mu_\ell
 \triangleq \mathcal{G}_\ell(\mu_{\mathrm{ch}}).
\label{eq:ovd_unperturbed_ga_mean}
\end{equation}
The corresponding genie-aided decision LLR is modeled as
\begin{equation}
 L^{\mathrm{G}}(u_{a_\ell})
 \mathrel{\overset{\mathrm{GA}}{\sim}}
 \mathcal{N}(\mu_\ell,2\mu_\ell).
\label{eq:ovd_unperturbed_ga_llr}
\end{equation}
Under the all-zero assumption, a negative decision LLR represents an
intrinsic error. Its GA probability is
\begin{equation}
 p_\ell
 \triangleq Q\!\left(\sqrt{\frac{\mu_\ell}{2}}\right),
\label{eq:ovd_unperturbed_ga_error}
\end{equation}
where $Q(\cdot)$ is the standard Gaussian tail function. 

For a perturbation branch with variance $v\geq 0$, adding independent
zero-mean Gaussian noise gives
\begin{equation}
 L_{\mathrm{ch}}^{(v)}
 \sim \mathcal{N}\!\left(
 \mu_{\mathrm{ch}},2\mu_{\mathrm{ch}}+v
 \right).
\label{eq:ovd_perturbed_channel}
\end{equation}
For $v>0$, this distribution is not consistent Gaussian because its
variance is not twice its mean. The following proposition gives a
decision-equivalent consistent Gaussian representation.

\begin{proposition}[Decision-equivalent Gaussian channel]
\label{prop:ovd_equivalent_channel}
Under  min-sum SC decoding,
the perturbed channel in \eqref{eq:ovd_perturbed_channel} is hard-decision
equivalent to a consistent Gaussian channel with mean
\begin{equation}
 \mu_{\mathrm{ch}}^{\mathrm{eq}}(v)
 \triangleq
 \frac{2\mu_{\mathrm{ch}}^2}{2\mu_{\mathrm{ch}}+v}.
\label{eq:ovd_equivalent_mean}
\end{equation}
This representation is obtained by scaling all perturbed channel LLRs
by a common positive factor, without changing the min-sum SC hard decisions.
\end{proposition}

\begin{IEEEproof}
For any $a>0$, the min-sum operations satisfy
\begin{equation*}
 f(ax,ay)=a f(x,y),
 \qquad
 g(ax,ay,u)=a g(x,y,u).
\label{eq:ovd_positive_homogeneity}
\end{equation*}
Suppose that all previous hard decisions are identical before and after scaling the channel LLRs by $a$. Then the same decision values
are used in the $g$-operations, and the LLR for the next decision is also scaled by $a$. Since $a>0$, its sign is unchanged, and hence the
corresponding hard decision remains unchanged. Applying this argument successively along the SC decoding, all hard decisions are preserved under a
common positive scaling of the channel LLRs. The same scaling invariance also holds for genie-aided SC decoding.

Choose
\begin{equation}
 \alpha(v)
 \triangleq
 \frac{2\mu_{\mathrm{ch}}}{2\mu_{\mathrm{ch}}+v}>0.
\label{eq:ovd_scaling_factor}
\end{equation}
The scaled channel LLR satisfies
\begin{equation}
 \alpha(v)L_{\mathrm{ch}}^{(v)}
 \sim \mathcal{N}\!\left(
 \alpha(v)\mu_{\mathrm{ch}},
 \alpha^2(v)(2\mu_{\mathrm{ch}}+v)
 \right).
\label{eq:ovd_scaled_distribution}
\end{equation}
By the definitions of $\alpha(v)$ and
$\mu_{\mathrm{ch}}^{\mathrm{eq}}(v)$,
\begin{equation}
 \begin{aligned}
 \alpha(v)\mu_{\mathrm{ch}}
 &=\mu_{\mathrm{ch}}^{\mathrm{eq}}(v),\\
 \alpha^2(v)(2\mu_{\mathrm{ch}}+v)
 &=2\mu_{\mathrm{ch}}^{\mathrm{eq}}(v).
 \end{aligned}
\label{eq:ovd_scaled_consistency}
\end{equation}
Hence,
\begin{equation}
 \alpha(v)L_{\mathrm{ch}}^{(v)}
 \sim \mathcal{N}\!\left(
 \mu_{\mathrm{ch}}^{\mathrm{eq}}(v),
 2\mu_{\mathrm{ch}}^{\mathrm{eq}}(v)
 \right).
\end{equation}
Thus, the scaled channel follows a consistent Gaussian distribution
and is hard-decision equivalent to the original perturbed channel
under min-sum SC decoding.
\end{IEEEproof}

Proposition~\ref{prop:ovd_equivalent_channel} establishes decision
equivalence between the perturbed channel and its scaled representation.
We therefore apply GA to the scaled channel and define
\begin{align}
 \mu_\ell(v)
 &\triangleq
 \mathcal{G}_\ell\!\left(
 \mu_{\mathrm{ch}}^{\mathrm{eq}}(v)
 \right),
\label{eq:ovd_perturbed_ga_mean}\\
 p_\ell(v)
 &\triangleq
 Q\!\left(\sqrt{\frac{\mu_\ell(v)}{2}}\right).
\label{eq:ovd_perturbed_ga_error}
\end{align}
Here, $\mu_\ell(v)$ is the GA mean of the decision LLR after perturbation
under the scaled representation, and $p_\ell(v)$ approximates the
corresponding error probability.

\subsubsection{Local Repair and Damage Probabilities}

For the $\ell$-th non-frozen position $a_\ell$,  $p_\ell$ and $p_\ell(v)$ describe the error probabilities
before and after perturbation, but do not specify how the decision
changes. We therefore introduce the repair probability $C_\ell(v)$ and
the damage probability $D_\ell(v)$.

Following the local perturbation approximation in \cite{b8}, let
$k_\ell$ denote the number of $g$-operations on the decoding path of
$a_\ell$. The effective perturbation variance at this position is
approximated by
\begin{equation}
 \nu_\ell(v)
 \triangleq
 2^{k_\ell}v.
\label{eq:ovd_local_variance}
\end{equation}
Under this approximation, the perturbation variance is unchanged through
an $f$-operation and doubled through a $g$-operation.

Let
\begin{equation}
 \Lambda_\ell
 \sim
 \mathcal{N}(\mu_\ell,2\mu_\ell)
\end{equation}
denote the unperturbed decision LLR under GA. To estimate the repair
probability, we introduce the local additive-Gaussian model
\begin{equation}
 \widetilde{\Lambda}_\ell(v)
 =
 \Lambda_\ell+Z_\ell(v),
 \qquad
 Z_\ell(v)
 \sim
 \mathcal{N}\!\left(0,\nu_\ell(v)\right),
\label{eq:ovd_local_transition_model}
\end{equation}
where $Z_\ell(v)$ is assumed independent of $\Lambda_\ell$.

Let $\varphi_\ell(x)$ denote the density of
$\mathcal{N}(\mu_\ell,2\mu_\ell)$, and let $\Phi(\cdot)$ denote the
standard Gaussian cumulative distribution function. For $v>0$, define
the repair probability as
\begin{align}
 C_\ell(v)
 &\triangleq
 \Pr\!\left(
 \widetilde{\Lambda}_\ell(v)>0
 \mid
 \Lambda_\ell<0
 \right)
\label{eq:ovd_repair_probability_definition}\\
 &=
 \frac{1}{p_\ell}
 \int_{-\infty}^{0}
 \Phi\!\left(
 \frac{x}{\sqrt{\nu_\ell(v)}}
 \right)
 \varphi_\ell(x)\,\mathrm{d}x.
\label{eq:ovd_repair_probability}
\end{align}
Obviously, $C_\ell(0)=0$.
Thus, $C_\ell(v)$ is the probability that an erroneous decision at position $a_\ell$ becomes correct after perturbation. 
Next, let $D_\ell(v)$ denote the probability that a correct decision at
position $a_\ell$ becomes erroneous after perturbation. To match the
perturbed marginal error probability $p_\ell(v)$, we require
\begin{equation}
 p_\ell(v)
 =
 p_\ell[1-C_\ell(v)]
 +(1-p_\ell)D_\ell(v).
\label{eq:ovd_probability_consistency}
\end{equation}
Hence,
\begin{equation}
 D_\ell(v)
 \triangleq
 \frac{
 p_\ell(v)-p_\ell[1-C_\ell(v)]
 }{
 1-p_\ell
 }.
\label{eq:ovd_damage_probability}
\end{equation}
Similarly, $D_\ell(0)=0$.

\begin{lemma}[Validity of repair and damage probabilities]
\label{lem:ovd_transition_validity}
For $v\geq0$,
\begin{equation}
 0\leq C_\ell(v)\leq\frac{1}{2},
 \qquad
 0\leq D_\ell(v)\leq\frac{1}{2}.
\label{eq:ovd_transition_bounds}
\end{equation}
Hence, $C_\ell(v)$ and $D_\ell(v)$ are valid transition probabilities.
\end{lemma}

\begin{IEEEproof}
For $v=0$, the result follows directly from
$C_\ell(0)=D_\ell(0)=0$. We therefore consider $v>0$.
Since
$\mu_{\mathrm{ch}}^{\mathrm{eq}}(v)
=2\mu_{\mathrm{ch}}^2/(2\mu_{\mathrm{ch}}+v)
\leq\mu_{\mathrm{ch}}$,
the GA recursion gives
$0\leq\mu_\ell(v)\leq\mu_\ell$.
Because the $Q$-function is decreasing on the nonnegative real line and
$Q(0)=1/2$, we have
\begin{equation}
 p_\ell
 \leq
 p_\ell(v)
 \leq
 \frac{1}{2}.
\label{eq:ovd_error_probability_bounds}
\end{equation}

Next, consider $C_\ell(v)$. For $x<0$ and $\nu_\ell(v)>0$,
$0\leq\Phi(x/\sqrt{\nu_\ell(v)})\leq1/2$.
Moreover,
$p_\ell=\int_{-\infty}^{0}\phi_\ell(x)\,dx$.
Therefore, from the definition of $C_\ell(v)$,
\[
 0\leq C_\ell(v)\leq\frac{1}{2}.
\]

Finally, from the definition of $D_\ell(v)$,
\[
 D_\ell(v)
 =
 \frac{p_\ell(v)-p_\ell+p_\ell C_\ell(v)}
 {1-p_\ell}.
\]
Since $p_\ell(v)\geq p_\ell$ and $C_\ell(v)\geq0$, we have
$D_\ell(v)\geq0$. Moreover, using
$p_\ell(v)\leq1/2$ and $C_\ell(v)\leq1/2$,
\begin{equation}
 D_\ell(v)
 \leq
 \frac{\frac{1}{2}-p_\ell+\frac{p_\ell}{2}}
 {1-p_\ell}
 =
 \frac{1}{2}.
\end{equation}
Hence, $0\leq D_\ell(v)\leq1/2$.
\end{IEEEproof}

If both $C_\ell(v)$ and $D_\ell(v)$ are evaluated directly from the local additive-Gaussian model, the resulting transition probabilities
do not in general satisfy the marginal consistency relation in \eqref{eq:ovd_probability_consistency}. Therefore, we use the local
Gaussian model only to estimate $C_\ell(v)$, while $D_\ell(v)$ is determined from \eqref{eq:ovd_probability_consistency} to preserve the
perturbed marginal error probability $p_\ell(v)$.

\subsubsection{Analytical Error Model and Recovery Recursions}

The marginal probabilities $p_\ell$ do not determine the joint
intrinsic-error pattern. For a  offline model, we approximate
the intrinsic-error indicators by independent Bernoulli variables:
\begin{equation}
 \widetilde{X}_\ell
 \mathrel{\overset{\mathrm{ind}}{\sim}}
 \operatorname{Bernoulli}(p_\ell),
 \qquad 0\leq\ell<M,
\label{eq:ovd_independent_states}
\end{equation}
where $\widetilde{X}_\ell=1$ indicates an analytical intrinsic error at
position $a_\ell$. These variables approximate the true error indicators
$X_\ell$ defined in \eqref{eq:ovd_genie_indicator}.

To simplify the analysis, we further assume that the repair and damage
transitions are independent across non-frozen positions. At position
$a_\ell$, an erroneous decision is repaired with probability
$C_\ell(v)$, while a correct decision remains correct with probability
$1-D_\ell(v)$. Note that these independence assumptions are used only in the
analytical model and approximate the actual SC recovery process.

We first characterize the number of errors after a given position.
Let $B_{\ell,j}$ denote the probability that exactly $j$ analytical
errors occur among the suffix positions
$a_{\ell+1},\ldots,a_{M-1}$. For the empty suffix,
\begin{equation}
 B_{M-1,0}=1.
\label{eq:ovd_B_boundary}
\end{equation}
For $\ell=M-2,\ldots,0$, the backward recursion is
\begin{equation}
 B_{\ell,j}
 =
 (1-p_{\ell+1})B_{\ell+1,j}
 +
 p_{\ell+1}B_{\ell+1,j-1}.
\label{eq:ovd_B_recursion}
\end{equation}
The first term corresponds to a correct decision at $a_{\ell+1}$,
whereas the second corresponds to an intrinsic error.

Next, for a given perturbation variance $v$, let $A_{\ell,j}(v)$ denote
the joint probability that the same suffix contains exactly $j$
analytical errors before perturbation and that all suffix decisions are
correct after perturbation. Its boundary condition is
\begin{equation}
 A_{M-1,0}(v)=1.
\label{eq:ovd_A_boundary}
\end{equation}
For $\ell=M-2,\ldots,0$,
\begin{equation}
 \begin{aligned}
 A_{\ell,j}(v)
 &=
 (1-p_{\ell+1})[1-D_{\ell+1}(v)]
 A_{\ell+1,j}(v)\\
 &\quad+
 p_{\ell+1}C_{\ell+1}(v)
 A_{\ell+1,j-1}(v).
 \end{aligned}
\label{eq:ovd_A_recursion}
\end{equation}
The first term represents an originally correct decision that remains
correct after perturbation. The second term represents an originally
erroneous decision that is repaired. 

Thus, $B_{\ell,j}$ gives the probability of $j$ original errors in the
suffix, whereas $A_{\ell,j}(v)$ additionally requires all suffix
decisions to be correct after perturbation.

Under the analytical model, the weight of class $(\ell,m)$ is
\begin{equation}
 w_{\ell,m}
 =
 \left[
 \prod_{s=0}^{\ell-1}(1-p_s)
 \right]
 p_\ell B_{\ell,m},
 \qquad
 (\ell,m)\in\mathcal{S}.
\label{eq:ovd_class_weight}
\end{equation}
The three factors correspond to no error before $a_\ell$, an intrinsic error at $a_\ell$, and exactly $m$ additional errors after $a_\ell$, respectively.
The analytical failure probability is
\begin{equation}
 \widetilde{P}_{\mathrm{fail}}
 \triangleq
 1-\prod_{s=0}^{M-1}(1-p_s).
\label{eq:ovd_analytical_failure_probability}
\end{equation}
For $\widetilde{P}_{\mathrm{fail}}>0$, the analytical class prior is defined as
\begin{equation}
 \widetilde{\pi}_{\ell,m}
 \triangleq
 \frac{w_{\ell,m}}
 {\widetilde{P}_{\mathrm{fail}}}.
\label{eq:ovd_analytical_prior}
\end{equation}

For a class with $B_{\ell,m}>0$, define the analytical single-branch
recovery probability as
\begin{equation}
 q_{\ell,m}(v)
 =
 C_\ell(v)
 \left[
 \prod_{s=0}^{\ell-1}[1-D_s(v)]
 \right]
 \frac{A_{\ell,m}(v)}
 {B_{\ell,m}}.
\label{eq:ovd_class_recovery}
\end{equation}
Here, $C_\ell(v)$ repairs the first intrinsic error,
$\prod_{s=0}^{\ell-1}[1-D_s(v)]$ keeps all preceding decisions correct,
and $A_{\ell,m}(v)/B_{\ell,m}$ is the conditional probability that the
entire suffix is correct after perturbation given exactly $m$ original
suffix errors. Empty products are defined as one.

Since $0\leq A_{\ell,j}(v)\leq B_{\ell,j}$, we have
\begin{equation*}
 0\leq q_{\ell,m}(v)\leq 1.
\end{equation*}
Moreover, $q_{\ell,m}(0)=0$. The probabilities
$\widetilde{\pi}_{\ell,m}$ and $q_{\ell,m}(v)$ are analytical
approximations to the true class priors and single-branch recovery
probabilities, respectively.

\subsection{Analytical Objective and Structural Property}
\label{subsec:analytical_objective_structure}

Since the exact frame-level objective is generally difficult to characterize within a fully analytical framework, we construct a tractable class-level
approximation. That is
\begin{equation}
 \widetilde{J}(V)
 =
 \sum_{(\ell,m)\in\mathcal{S}}
 \widetilde{\pi}_{\ell,m}
 \left\{
 1-
 \prod_{v_t\in V}
 \left[1-q_{\ell,m}(v_t)\right]
 \right\}.
\label{eq:ovd_analytical_objective}
\end{equation}
Equation~\eqref{eq:ovd_analytical_objective} has the same coverage structure as the exact objective in \eqref{eq:ovd_true_objective}, but it is not an exact
representation of $J(V)$. First, the frame-level conditional expectation is replaced by a product of class-level recovery probabilities. Second, the class
probabilities and recovery curves are obtained under Gaussian approximation, effective-perturbation, and independence approximations. Hence,
$\widetilde{J}(V)$ is used as a tractable analytical approximation to $J(V)$.

Let $V\uplus\{v\}$ denote the perturbation variance multiset obtained by adding one independent perturbation branch with variance $v$ to the current
multiset $V$.  Define the exact and analytical marginal gains as
\begin{align}
 \Delta(v\mid V)
 &\triangleq
 J\!\left(V\uplus\{v\}\right)-J(V),
\label{eq:ovd_true_set_marginal_definition}\\
 \widetilde{\Delta}(v\mid V)
 &\triangleq
 \widetilde{J}\!\left(V\uplus\{v\}\right)-\widetilde{J}(V).
\label{eq:ovd_analytical_set_marginal_definition}
\end{align}

\begin{theorem}[Monotonicity and diminishing marginal recovery gains]
\label{thm:ovd_objective_structure}
The exact objective $J$ and the analytical objective $\widetilde{J}$ satisfy
\begin{equation}
 J(\varnothing)=0,
 \qquad
 \widetilde{J}(\varnothing)=0.
\label{eq:ovd_set_normalization}
\end{equation}

For any variance multiset $V$ and any additional variance $v$,
\begin{equation}
 J\!\left(V\uplus\{v\}\right)\geq J(V),
 \qquad
 \widetilde{J}\!\left(V\uplus\{v\}\right)\geq\widetilde{J}(V).
\label{eq:ovd_set_monotonicity}
\end{equation}

Moreover, for any two variances $v$ and $v'$,
\begin{equation}
 0\leq
 \Delta\!\left(v\mid V\uplus\{v'\}\right)
 \leq
 \Delta(v\mid V),
\label{eq:ovd_true_set_diminishing}
\end{equation}
and
\begin{equation}
 0\leq
 \widetilde{\Delta}\!\left(v\mid V\uplus\{v'\}\right)
 \leq
 \widetilde{\Delta}(v\mid V).
\label{eq:ovd_analytical_set_diminishing}
\end{equation}

Thus, adding one independent perturbation branch cannot reduce either recovery
objective, while the marginal recovery gain of another branch cannot increase
after one more branch has already been added.
\end{theorem}

\begin{IEEEproof}
For $V=\varnothing$, the failure-probability product is an empty product and
therefore equals one. Hence, no perturbation branch is available to recover
the frame, which gives \eqref{eq:ovd_set_normalization}.

From the definition of the exact objective, the marginal gain of adding one
branch with variance $v$ is
\begin{equation}
 \Delta(v\mid V)
 =
 \mathbb{E}_{\omega}\!\left[
 s_{\omega}(v)
 \prod_{v_t\in V}
 \left[1-s_{\omega}(v_t)\right]
 \,\middle|\,
 \mathrm{Fail}_{\mathrm{SC}}
 \right].
\label{eq:ovd_true_set_marginal}
\end{equation}
For a given failed frame, the product in \eqref{eq:ovd_true_set_marginal} gives the probability that all branches in
$V$ fail, while $s_{\omega}(v)$ gives the recovery probability of the newly added branch.

Similarly, the analytical marginal gain is
\begin{equation}
 \begin{aligned}
 \widetilde{\Delta}(v\mid V)
 &=
 \sum_{(\ell,m)\in\mathcal{S}}
 \widetilde{\pi}_{\ell,m}q_{\ell,m}(v)\\
 &\quad\times
 \prod_{v_t\in V}
 \left[1-q_{\ell,m}(v_t)\right].
 \end{aligned}
\label{eq:ovd_analytical_set_marginal}
\end{equation}
All factors in these expressions lie in $[0,1]$. Therefore,
\[
 \Delta(v\mid V)\geq 0,
 \qquad
 \widetilde{\Delta}(v\mid V)\geq 0,
\]
which proves \eqref{eq:ovd_set_monotonicity}.

Now suppose that one branch with variance $v'$ has already been added. Then
\begin{equation}
 \begin{aligned}
 \Delta\!\left(v\mid V\uplus\{v'\}\right)
 &=
 \mathbb{E}_{\omega}\!\left[
 s_{\omega}(v)
 \left[1-s_{\omega}(v')\right]\right.\\
 &\qquad\left.
 \times
 \prod_{v_t\in V}
 \left[1-s_{\omega}(v_t)\right]
 \,\middle|\,
 \mathrm{Fail}_{\mathrm{SC}}
 \right].
 \end{aligned}
\label{eq:ovd_true_set_diminishing_proof}
\end{equation}
Since
\[
 0\leq 1-s_{\omega}(v')\leq 1,
\]
we have
\[
 \Delta\!\left(v\mid V\uplus\{v'\}\right)
 \leq
 \Delta(v\mid V).
\]

Similarly,
\begin{equation}
 \begin{aligned}
 \widetilde{\Delta}\!\left(v\mid V\uplus\{v'\}\right)
 &=
 \sum_{(\ell,m)\in\mathcal{S}}
 \widetilde{\pi}_{\ell,m}q_{\ell,m}(v)
 \left[1-q_{\ell,m}(v')\right]\\
 &\quad\times
 \prod_{v_t\in V}
 \left[1-q_{\ell,m}(v_t)\right].
 \end{aligned}
\label{eq:ovd_analytical_set_diminishing_proof}
\end{equation}
Since
\[
 0\leq 1-q_{\ell,m}(v')\leq 1,
\]
we obtain
\[
 \widetilde{\Delta}\!\left(v\mid V\uplus\{v'\}\right)
 \leq
 \widetilde{\Delta}(v\mid V).
\]

Together with the nonnegativity of both marginal gains, this proves
\eqref{eq:ovd_true_set_diminishing} and
\eqref{eq:ovd_analytical_set_diminishing}.
\end{IEEEproof}

\section{Offline Variance Design for PSCP Decoding}
\label{sec:ovd_design}

Based on the analytical objective in Section~III, this section constructs a finite candidate variance set and selects $T$ perturbation variances using a greedy algorithm.
The variance design is performed offline for a fixed code, operating SNR, and number of perturbation branches.
We also present the complete OVD-PSCP procedure, which combines offline variance design with online parallel decoding.

\subsection{Candidate Variance Set}
\label{subsec:ovd_candidate_set}

To construct a finite set of candidate perturbation variances, we first
characterize the perturbation strength at each non-frozen position. Under
the local model in \eqref{eq:ovd_local_variance}, define
\begin{equation}
 \eta_\ell(v)
 \triangleq
 \frac{\nu_\ell(v)}{2\mu_\ell+\nu_\ell(v)}
 =
 \frac{v}{\rho_\ell+v},
 \qquad
 \rho_\ell
 \triangleq
 \frac{2\mu_\ell}{2^{k_\ell}}.
\label{eq:ovd_sensitivity_ratio}
\end{equation}
Here, $\eta_\ell(v)$ represents the fraction of the local Gaussian variance
contributed by the perturbation. It increases monotonically from $0$ to $1$
as $v$ increases, while $\rho_\ell$ determines the variance scale of this
transition.

For a given $0<\epsilon<1/2$, the values
\[
 \rho_\ell\frac{\epsilon}{1-\epsilon}
 \quad\text{and}\quad
 \rho_\ell\frac{1-\epsilon}{\epsilon}
\]
correspond to $\eta_\ell(v)=\epsilon$ and
$\eta_\ell(v)=1-\epsilon$, respectively. We therefore use this interval to
represent the main transition region for position $a_\ell$.

Since different non-frozen positions have different $\rho_\ell$, define
\begin{equation}
 \rho_{\min}
 \triangleq
 \min_{0\leq\ell<M}\rho_\ell,
 \qquad
 \rho_{\max}
 \triangleq
 \max_{0\leq\ell<M}\rho_\ell,
\label{eq:ovd_rho_range}
\end{equation}
and choose the common design range
\begin{equation}
 v_{\min}
 \triangleq
 \rho_{\min}\frac{\epsilon}{1-\epsilon},
 \qquad
 v_{\max}
 \triangleq
 \rho_{\max}\frac{1-\epsilon}{\epsilon}.
\label{eq:ovd_variance_range}
\end{equation}
This range covers the main transition regions of all non-frozen positions under the local model.

Let $G$ denote the number of candidate variances. For $G\geq2$, define
\begin{equation}
  V_{\mathrm c}
 \triangleq
 \{v_0,v_1,\ldots,v_{G-1}\},
\label{eq:ovd_candidate_set}
\end{equation}
where
\begin{equation}
 v_g
 \triangleq
 v_{\min}
 \left(
 \frac{v_{\max}}{v_{\min}}
 \right)^{g/(G-1)},
 \qquad
 g=0,1,\ldots,G-1.
\label{eq:ovd_log_grid}
\end{equation}
Thus, the candidate variances are logarithmically spaced over $[v_{\min},v_{\max}]$.
The interval in \eqref{eq:ovd_variance_range} is only a model-based design range and need not contain the continuous-domain optimum.

\begin{algorithm}[t]
\caption{OVD-PSCP Decoding}
\label{alg:ovd_greedy}
\begin{algorithmic}[1]

\STATE \textit{// Offline variance design}
\STATE \textbf{Input:} Candidate set $\mathcal V_{\mathrm c}$,
$\{q_{\ell,m,g}\}$, $\{\widetilde{\pi}_{\ell,m}\}$, and $T$;
\STATE \textbf{Output:} Variance multiset $V$;
\STATE $V\leftarrow\varnothing$;
\STATE $\xi_{\ell,m}\leftarrow1$ for all
$(\ell,m)\in\mathcal S$;

\FOR{$\tau=0$ \TO $T-1$}
    \FOR{$g=0$ \TO $G-1$}
        \STATE Compute $\widetilde{\Delta}(v_g\mid V)$ using
        \eqref{eq:ovd_candidate_marginal_gain};
    \ENDFOR
    \STATE $g^\star\leftarrow
    \mathop{\arg\max}_{0\leq g<G}
    \widetilde{\Delta}(v_g\mid V)$;
    \STATE $V\leftarrow V\uplus\{v_{g^\star}\}$;
    \FORALL{$(\ell,m)\in\mathcal S$}
        \STATE $\xi_{\ell,m}\leftarrow
        \xi_{\ell,m}(1-q_{\ell,m,g^\star})$;
    \ENDFOR
\ENDFOR

\STATE \textit{// Online parallel decoding}
\STATE \textbf{Input:} Channel LLR vector $\mathbf{L}(y)$,
stored variance multiset $V=\{v_0,\ldots,v_{T-1}\}$,
and non-frozen set $\mathcal I$;
\STATE \textbf{Output:} Decoded estimate $\hat{\mathbf{u}}$
or decoding failure;

\FOR{$b=0$ \TO $T$ \textbf{in parallel}}
    \IF{$b=0$}
        \STATE $\hat{\mathbf{u}}^{(b)}\leftarrow
        \mathrm{SCDecoder}(\mathbf{L}(y),\mathcal I)$;
    \ELSE
        \STATE Generate $\mathbf{L}^{(b-1)}(p)$ using
        \eqref{eq:parallel_perturbed_llr} with variance
        $v_{b-1}$;
        \STATE $\hat{\mathbf{u}}^{(b)}\leftarrow
        \mathrm{SCDecoder}(\mathbf{L}^{(b-1)}(p),\mathcal I)$;
    \ENDIF
\ENDFOR

\FOR{$b=0$ \TO $T$}
    \IF{$\mathrm{CRCCheck}(\hat{\mathbf{u}}^{(b)})
        =\mathrm{success}$}
        \RETURN $\hat{\mathbf{u}}^{(b)}$;
    \ENDIF
\ENDFOR
\RETURN failure;

\end{algorithmic}
\end{algorithm}

\subsection{Greedy Variance Allocation}
\label{subsec:ovd_greedy}

For each candidate variance $v_g\in V_{\mathrm c}$, precompute
\begin{equation}
 q_{\ell,m,g}
 \triangleq
 q_{\ell,m}(v_g),
 \qquad
 (\ell,m)\in\mathcal S,
 \quad
 0\leq g<G.
\label{eq:ovd_candidate_recovery_table}
\end{equation}
For a variance multiset $V$, define
\begin{equation}
 \xi_{\ell,m}(V)
 \triangleq
 \prod_{v_t\in V}
 \left[1-q_{\ell,m}(v_t)\right],
\label{eq:ovd_class_residual}
\end{equation}
which is the probability that class $(\ell,m)$ remains unrecovered by all branches in $V$.

From the analytical marginal gain defined in Section~III, adding one branch with candidate variance $v_g$ gives
\begin{align}
 \widetilde{\Delta}(v_g\mid V)
 &=
 \widetilde J(V\uplus\{v_g\})
 -\widetilde J(V)
\label{eq:ovd_candidate_marginal_gain_definition}\\
 &=
 \sum_{(\ell,m)\in\mathcal S}
 \widetilde{\pi}_{\ell,m}
 q_{\ell,m,g}
 \xi_{\ell,m}(V).
\label{eq:ovd_candidate_marginal_gain}
\end{align}

Based on Theorem~\ref{thm:ovd_objective_structure}, the variance multiset is constructed greedily.
Starting from $V=\varnothing$, the candidate variance with the largest current marginal gain in the analytical objective is added at each step.
The same candidate variance may be selected multiple times. After $T$ selections, the variances in $V$ are assigned to the $T$ perturbation branches in selection order.
The variance design is performed offline, and the selected variances are stored for reuse in online decoding.
The complete OVD-PSCP procedure, including offline variance design and online parallel decoding, is summarized in Algorithm~\ref{alg:ovd_greedy}.

The offline complexity consists of the analytical precomputation and the greedy variance allocation. The $B_{\ell,j}$ recursion is evaluated once
and requires $O(M^2)$ operations. For each of the $G$ candidate variances, one GA evaluation and one $A_{\ell,j}(v_g)$ recursion are required. With an
$O(N)$ GA implementation, these computations require $O(GN+GM^2)$ operations in total. During greedy allocation, the marginal
gain of each candidate is evaluated over $|\mathcal S|=M(M+1)/2$ error classes. Hence, each greedy selection requires
$O(GM^2)$ operations, and $T$ selections require $O(TGM^2)$ operations. Therefore, the overall offline computational complexity is
$O(GN+TGM^2)$, while storing all class-wise recovery probabilities $\{q_{\ell,m,g}\}$ requires $O(GM^2)$ memory.

\section{Simulation Results}
In this section, simulation results are provided to show the impact of the proposed OVD-PSCP decoding on the BLER performance.
The codes are designed according to the 5G standard \cite{b2}.
the parameter $\epsilon$ in section IV is set to $0.02$, and the size of  $V_{\mathrm c}$ is set to $G = 2049$.
A frame is considered correctly decoded if it passes the CRC check.

For SCP, the perturbation variance follows the setting in \cite{b11} and is given by
\begin{equation}
    \sigma_p^2
    =
    \frac{1}{2R}
    10^{-\frac{\mathrm{SNR}-0.5}{10}}
    -
    \sigma^2.
\end{equation}
Note that the perturbation variance in \cite{b11} is applied to the received vector, whereas the perturbation variance in the proposed scheme is applied to the LLRs.

\subsection{Error-Correction Performance}

Fig.~1 compares the BLER performance of SCP and OVD-PSCP for $\mathcal{P}(64,32+6)$, where CA-SCL decoding with $L=4$ is included as a reference. 
The CRC polynomial is ${\rm CRC}(x)=x^6+x^5+1$.
For the same number of $T$, OVD-PSCP consistently outperforms SCP, and the performance gain becomes more pronounced as $T$ increases.
In particular, when $T=80$, OVD-PSCP achieves approximately $0.4$--$0.6$~dB gain over SCP in the BLER range from $10^{-3}$ to $10^{-4}$.
Moreover, OVD-PSCP with approximately $T=20$ achieves performance comparable to CA-SCL decoding with $L=4$.

\begin{figure}[t]
  \centering
  \includegraphics[width=1\linewidth]{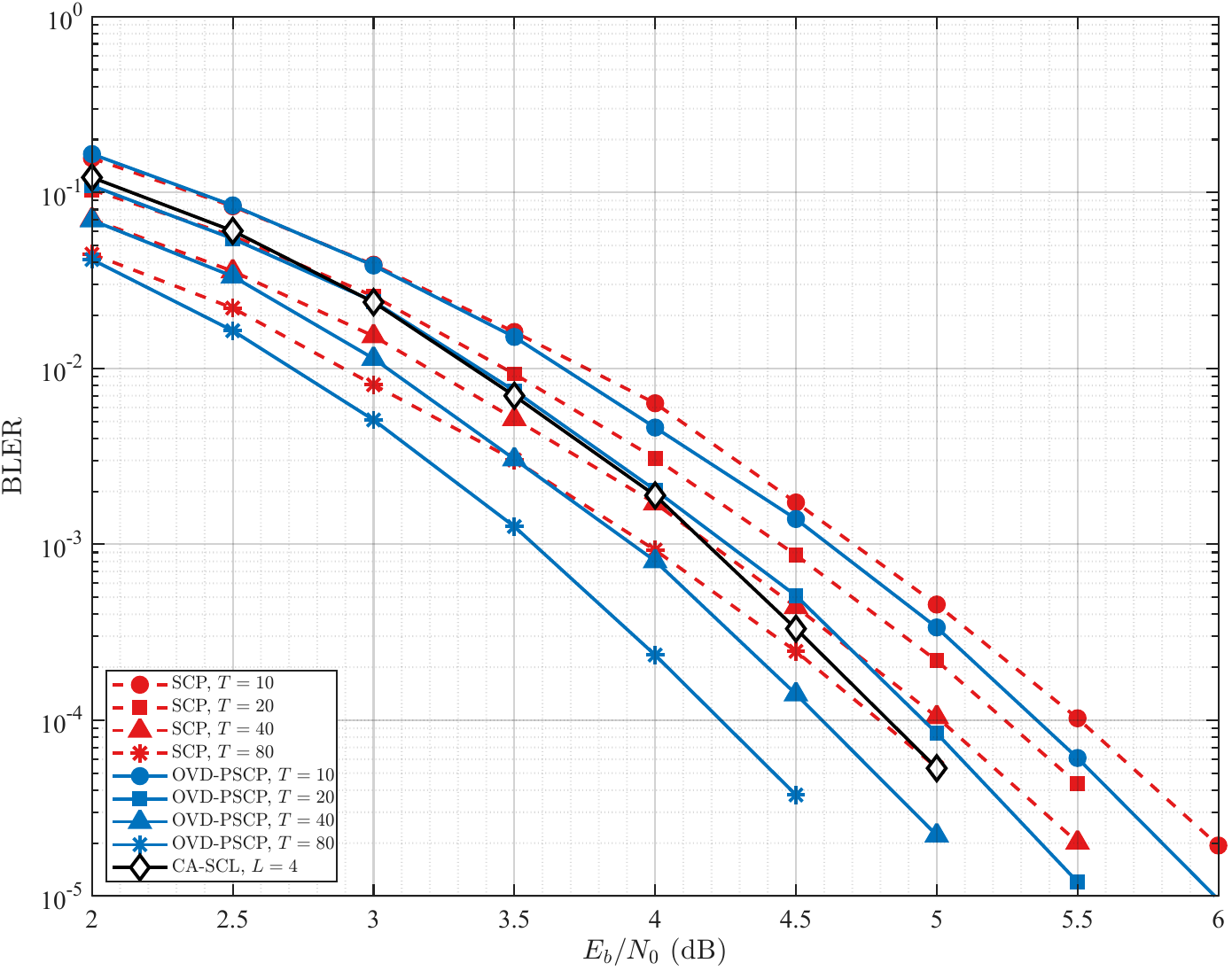}
  \caption{Performance comparison with $\mathcal{P}(64,32+6)$.}
  \label{fig:N64}
\end{figure}

\begin{figure}[t]
  \centering
  \includegraphics[width=1\linewidth]{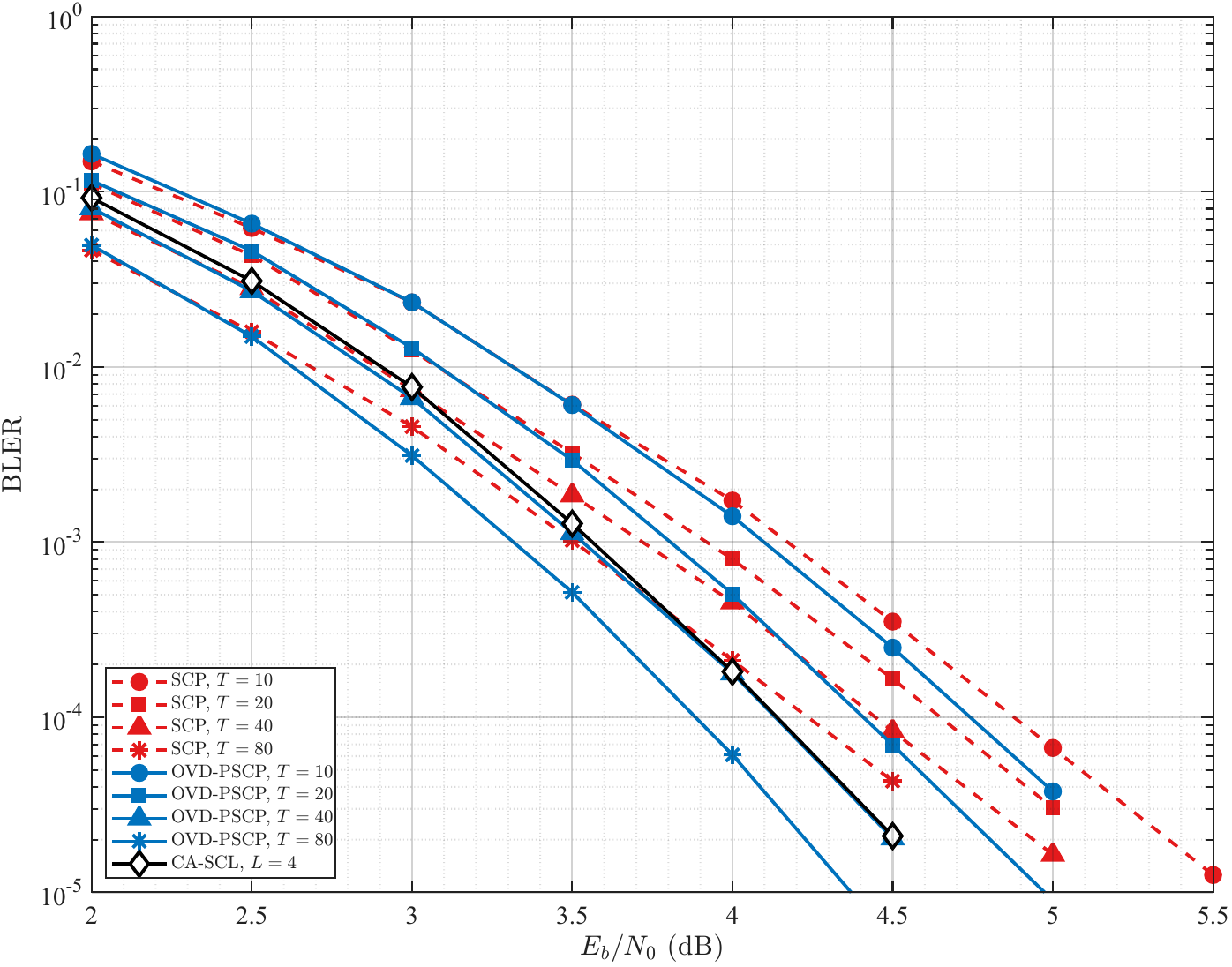}
  \caption{Performance comparison with $\mathcal{P}(128,64+8)$.}
  \label{fig:N128}
\end{figure}

\begin{figure}[t]
  \centering
  \includegraphics[width=1\linewidth]{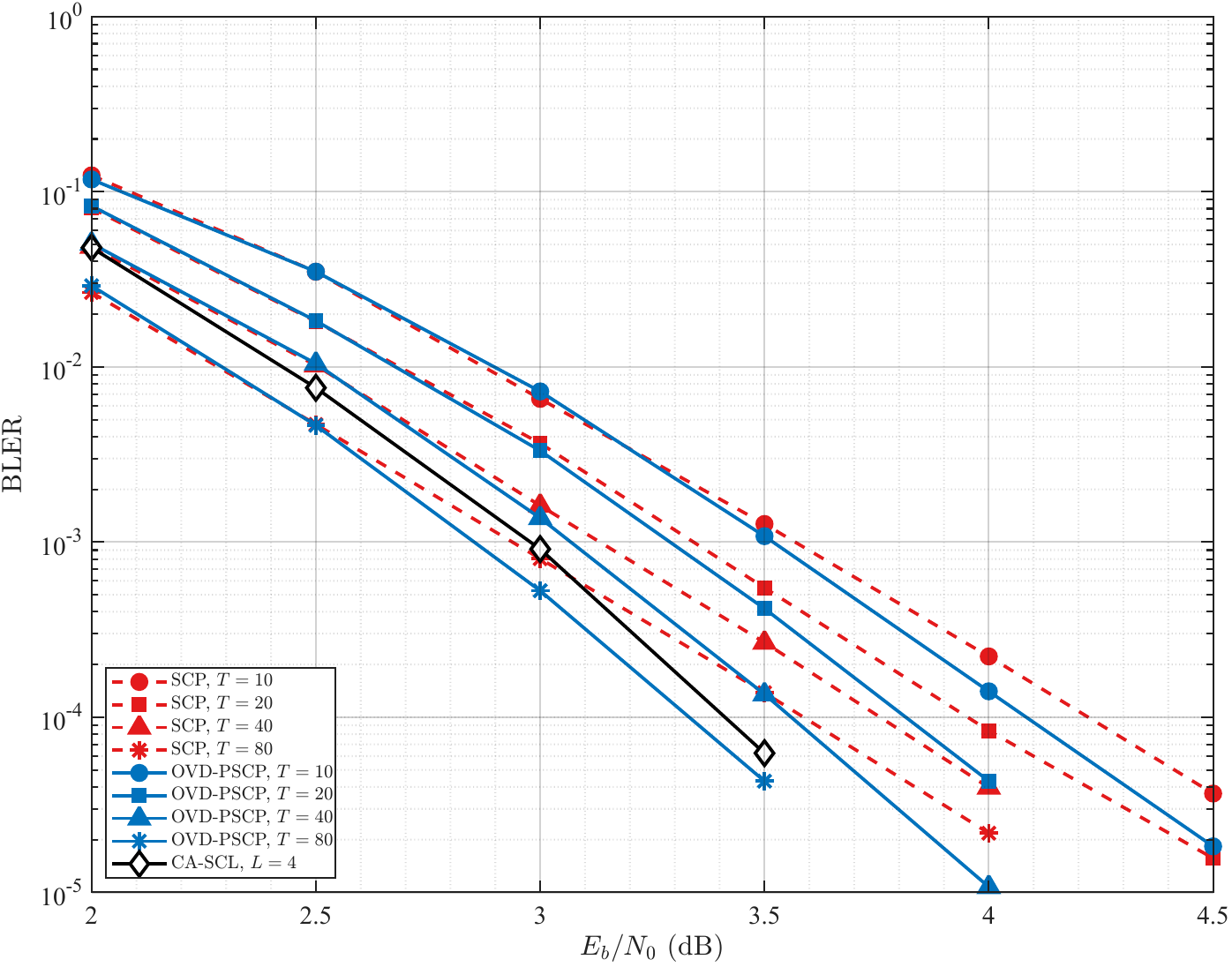}
  \caption{Performance comparison with  $\mathcal{P}(256,128+8)$.}
  \label{fig:N256}
\end{figure}

\begin{figure}[t]
  \centering
  \includegraphics[width=1\linewidth]{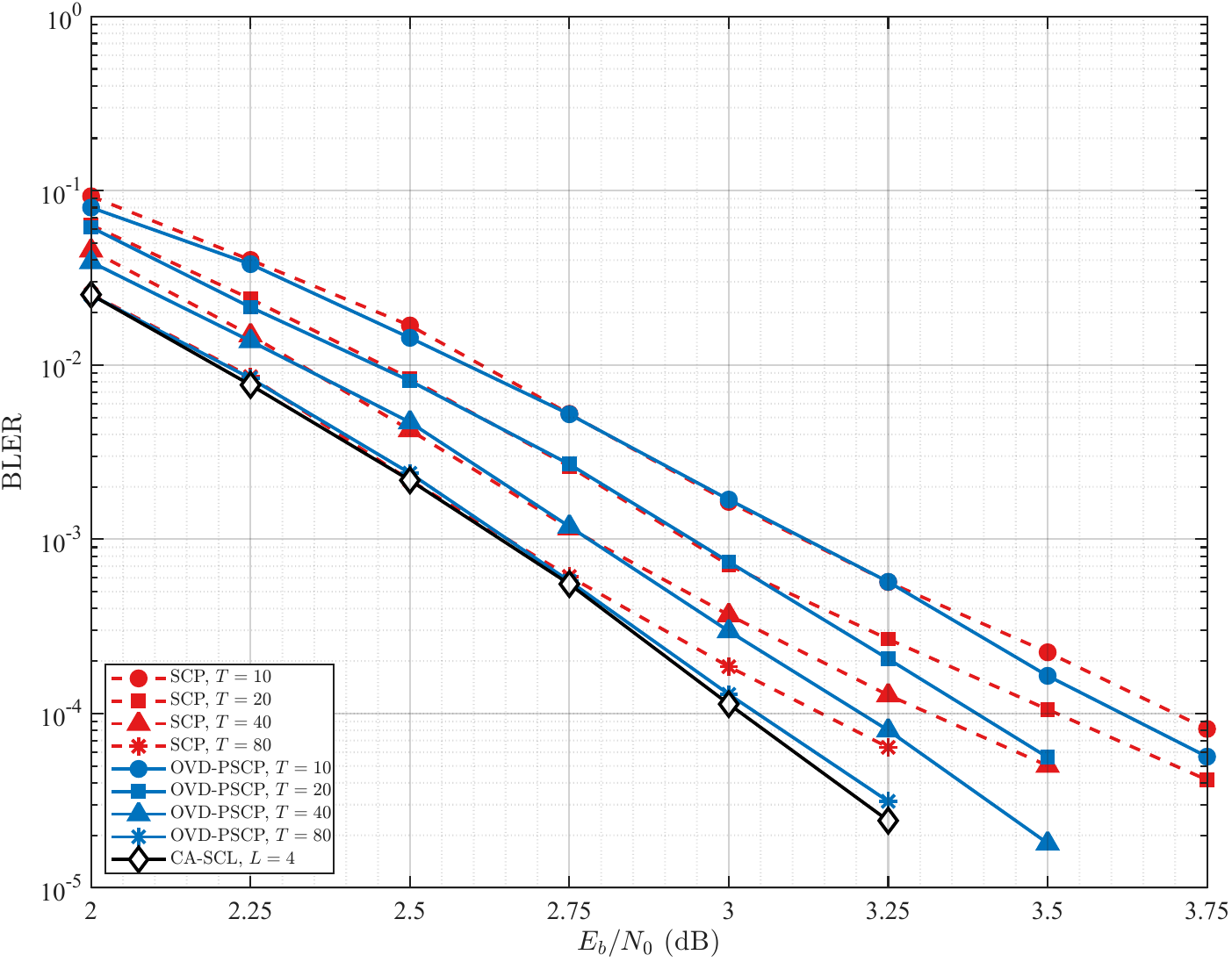}
  \caption{Performance comparison with $\mathcal{P}(512,256+11)$.}
  \label{fig:N512}
\end{figure}

Fig.~2 compares the BLER performance of SCP and OVD-PSCP for $\mathcal{P}(128,64+8)$, where CA-SCL decoding with $L=4$ is included as a reference. 
The CRC polynomial is ${\rm{CRC}}(x) = x^8 + x^2 + x + 1$.
For the same number of $T$, OVD-PSCP consistently outperforms SCP, and the performance gain becomes more pronounced as $T$ increases.
In particular, when $T=80$, OVD-PSCP achieves approximately $0.2$--$0.35$~dB gain over SCP in the BLER range from $10^{-3}$ to $10^{-4}$.
Moreover, OVD-PSCP with approximately $T=40$ achieves performance comparable to CA-SCL decoding with $L=4$.

Fig.~3 compares the BLER performance of SCP and OVD-PSCP for $\mathcal{P}(256,128+8)$, where CA-SCL decoding with $L=4$ is included as a reference. 
The CRC polynomial is ${\rm{CRC}}(x) = x^8 + x^2 + x + 1$.
For the same number of $T$, OVD-PSCP consistently outperforms SCP, and the performance gain becomes more pronounced as $T$ increases.
In particular, when $T=80$, OVD-PSCP achieves approximately $0.1$--$0.25$~dB gain over SCP in the BLER range from $10^{-3}$ to $10^{-4}$.
Moreover, OVD-PSCP with approximately $T=60$ achieves performance comparable to CA-SCL decoding with $L=4$.

Fig.~4 compares the BLER performance of SCP and OVD-PSCP for $\mathcal{P}(512,256+11)$, where CA-SCL decoding with $L=4$ is included as a reference.
The CRC polynomial is ${\rm{CRC}}(x) = {x^{11}} + {x^9}  + 1$.
For the same number of $T$, OVD-PSCP consistently outperforms SCP, and the performance gain becomes more pronounced as $T$ increases.
In particular, when $T=80$, OVD-PSCP achieves approximately $0.1$~dB gain over SCP at the BLER of $10^{-4}$.
Moreover, OVD-PSCP with approximately $T=80$ achieves performance comparable to CA-SCL decoding with $L=4$.

The above results show that OVD-PSCP achieves larger gains over SCP for short- and moderate-length polar codes, while its advantage
decreases as the code length increases. Two factors may help explain this trend.
First, the asymptotic analysis in \cite{b8} shows that, under its assumptions, the probability of perturbation introducing an error
before the original first error position tends to zero as the code length increases. 
The proposed method selects perturbation variances by considering both the probability of correcting an error, $C_{\ell}(v)$,
and the probability of introducing a new error, $D_{\ell}(v)$.
If perturbation-induced damage has a smaller overall effect for longer codes, the additional gain from reducing such damage through variance selection may also become smaller.
Second, the analytical objective in \eqref{eq:ovd_analytical_objective} uses a single recovery probability for each error class at a given perturbation variance. 
However, recovery probabilities vary across frames within the same error class.
For longer codes with stronger channel polarization, the remaining SC-failed frames
tend to be harder to recover through perturbation. The class-level approximation does not distinguish these difficult frames from easier ones in the same class, 
which can limit the effectiveness of variance selection.

\subsection{Candidate Codeword Analysis}

\begin{table*}[t]
\centering
\caption{$\mathcal{P}(64,32+6)$ Candidate Codeword Statistics}
\label{tab:n64-codeword-statistics}

{\small
\renewcommand{\arraystretch}{1.12}
\begin{tabular*}{\textwidth}
{@{\extracolsep{\fill}}ccc*{7}{c}@{}}
\toprule
\multirow{2}{*}{Perturbations}
& \multirow{2}{*}{Metric}
& \multirow{2}{*}{Method}
& \multicolumn{7}{c}{Eb/N0 (dB)} \\
\cmidrule(lr){4-10}
& & & 2.0 & 2.5 & 3.0 & 3.5 & 4.0 & 4.5 & 5.0 \\
\midrule

\multirow{4}{*}{$T=10$}
& \multirow{2}{*}{DC}
& SCP & 60.62\% & 55.99\% & 51.05\% & 45.02\% & 39.08\% & 33.03\% & 27.71\% \\
& & OVD-PSCP & 51.53\% & 53.05\% & 54.91\% & 57.04\% & 58.45\% & 60.13\% & 61.80\% \\
\cmidrule(lr){2-10}
& \multirow{2}{*}{EML}
& SCP & 12.76\% & 11.88\% & 10.89\% & 9.78\% & 8.25\% & 7.00\% & 5.79\% \\
& & OVD-PSCP & 12.10\% & 11.52\% & 11.15\% & 10.75\% & 9.47\% & 8.60\% & 7.69\% \\
\midrule

\multirow{4}{*}{$T=20$}
& \multirow{2}{*}{DC}
& SCP & 53.43\% & 48.31\% & 43.01\% & 37.24\% & 31.33\% & 25.66\% & 20.90\% \\
& & OVD-PSCP & 46.27\% & 48.77\% & 49.91\% & 52.33\% & 53.62\% & 55.68\% & 57.08\% \\
\cmidrule(lr){2-10}
& \multirow{2}{*}{EML}
& SCP & 9.84\% & 9.07\% & 8.20\% & 7.03\% & 5.83\% & 4.75\% & 3.84\% \\
& & OVD-PSCP & 9.51\% & 9.18\% & 8.54\% & 8.11\% & 7.38\% & 6.82\% & 6.18\% \\
\midrule

\multirow{4}{*}{$T=40$}
& \multirow{2}{*}{DC}
& SCP & 44.99\% & 39.85\% & 34.23\% & 28.96\% & 23.70\% & 18.79\% & 14.83\% \\
& & OVD-PSCP & 42.67\% & 45.08\% & 47.62\% & 49.55\% & 51.72\% & 54.13\% & 55.53\% \\
\cmidrule(lr){2-10}
& \multirow{2}{*}{EML}
& SCP & 7.07\% & 6.41\% & 5.49\% & 4.69\% & 3.69\% & 2.95\% & 2.24\% \\
& & OVD-PSCP & 6.75\% & 6.78\% & 6.34\% & 5.93\% & 5.36\% & 5.17\% & 4.48\% \\
\midrule

\multirow{4}{*}{$T=80$}
& \multirow{2}{*}{DC}
& SCP & 35.16\% & 30.83\% & 26.31\% & 21.63\% & 17.10\% & 13.16\% & 10.02\% \\
& & OVD-PSCP & 38.95\% & 41.46\% & 44.31\% & 46.41\% & 48.90\% & 50.76\% & 52.62\% \\
\cmidrule(lr){2-10}
& \multirow{2}{*}{EML}
& SCP & 4.52\% & 4.01\% & 3.50\% & 2.93\% & 2.28\% & 1.72\% & 1.28\% \\
& & OVD-PSCP & 4.72\% & 4.49\% & 4.46\% & 4.00\% & 3.76\% & 3.38\% & 3.13\% \\
\bottomrule
\end{tabular*}
}
\end{table*}

To facilitate a fair comparison between the proposed scheme and the typical scheme,
based on 10000 error frames, we collect the following statistics: (i) the proportion of distinct codewords among all candidate codewords generated over all perturbation branches,
and (ii) the proportion of perturbation branches in which the maximum likelihood (ML) metric \cite{b13} of distinct codewords exceeds that of the codeword produced by the original SC decoder. For convenience, we denote the result in (i) by DC, and denote the result in (ii) by EML.

DC reflects the diversity of candidates available for CRC checking. EML indicates how often perturbation produces
candidates that are more likely to be correct than the original SC output, as judged by the ML metric.

Table~I presents the candidate codeword statistics for $\mathcal{P}(64,32+6)$ under different SNRs and perturbation branches. 
As the SNR increases, the DC of SCP decreases, whereas that of OVD-PSCP increases. At SNRs of $3$~dB and above, OVD-PSCP achieves higher DC and EML than SCP for all listed
values of $T$. For example, with $T=80$ at an SNR of $5$~dB, OVD-PSCP achieves a DC of $52.62\%$ and an EML of $3.13\%$,
compared with $10.02\%$ and $1.28\%$, respectively, for SCP. 
At lower SNRs, the comparison between the two schemes depends on the number of perturbation attempts. For example, at an
SNR of $2$~dB, OVD-PSCP yields lower DC and EML than SCP for $T=10$, $20$, and $40$. With $T=80$, however, OVD-PSCP
achieves higher values for both metrics: its DC and EML reach $38.95\%$ and $4.72\%$, respectively, compared with
$35.16\%$ and $4.52\%$ for SCP.
Overall, these statistics show that OVD-PSCP generates more distinct candidates and more frequently improves the ML metric at higher SNRs, which is
consistent with the BLER gains shown in Fig.~1.

\section{Conclusion}

In this paper, we proposed an offline variance design method for PSCP decoding of short- and medium-length polar codes. An analytical recovery model was developed to guide the selection of perturbation variances, and a greedy algorithm was used to determine the variances for a fixed number of branches. Since all variances are determined before decoding, the original SC branch and all perturbation branches can start simultaneously.

Simulation results for polar codes with lengths $64$, $128$, $256$, and $512$ showed that OVD-PSCP achieves lower BLER than conventional SCP with the same number of perturbation branches. The improvement is more pronounced for shorter codes and larger numbers of perturbation branches. 
Future work will study the recovery difficulty of SC-failed frames within the same error class at different SNRs and code lengths. This may provide more detailed guidance for perturbation variance design.


\begin{thebibliography}{00}
\bibitem{b1} 	E. Arıkan, “Channel polarization: A method for constructing capacityachieving codes for symmetric binary-input memoryless channels,” \emph{IEEE Trans. Inf. Theory}, vol. 55, no. 7, pp. 3051–3073, Jul. 2009.
\bibitem{b2} 	3rd Generation Partnership Project (3GPP), ``NR; Multiplexing and channel coding,'' 3GPP TS 38.212, Release 15, 2018.
\bibitem{b3}	J. Tong, X. Wang, Q. Zhang, H. Zhang, R. Li, and J. Wang, “Fast Polar Codes for Terabits-Per-Second Throughput Communications,” \emph{Communications of HUAWEI RESEARCH}, pp. 110–121, Sep. 2022.
\bibitem{b4} 	I. Tal and A. Vardy, “List decoding of polar codes,” \emph{IEEE Trans. Inf. Theory}, vol. 61, no. 5, pp. 2213–2226, May 2015.
\bibitem{b5} 	K. Niu and K. Chen, “CRC-aided decoding of polar codes,” \emph{IEEE Commun. Lett.}, vol. 16, no. 10, pp. 1668–1671, Oct. 2012.
\bibitem{b6}	G. Nana, S. Zhao, and L. Kong, “CRC-aided perturbed decoding of polar codes,” in \emph{Int. Conf. on Wireless Commun., Networking and Mobile Comput. (WiCOM)}, 2018.
\bibitem{b7} 	D. Xiao and Z. Gu, “Dynamic perturbation decoding of polar-CRC cascaded code,” in \emph{Int. Wireless Commun. and Mobile Comput. (IWCMC)}, 2020.
\bibitem{b8} 	Z. Liu, L. Yao, S. Yuan, G. Yan, Z. Ma, and Y. Liu, “Performance Analysis of Perturbation-Enhanced SC Decoders,” \emph{IEEE Commun. Lett.}, vol. 29, no. 3, pp. 507–511, Mar. 2025.
\bibitem{b9}	C. Pillet, I. Sagitov, D. Deslandes, and P. Giard, “Successive-Cancellation Flip and Perturbation Decoder of Polar Codes,” in \emph{Proc. 2025 IEEE Wireless Communications and Networking Conference (WCNC)}, 2025.
\bibitem{b10} 	Z. Yang, L. Chen, K. Qin, X. Wang, and H. Zhang, “Perturbation-Based Decoding Schemes for Long Polar Codes,” in \emph{Proc. 2025 IEEE International Symposium on Information Theory (ISIT)}, 2025.
\bibitem{b11}    Z. Yang, L. Chen, K. Qin, X. Wang, and H. Zhang, ``Improved successive cancellation decoding of polar codes through  perturbing a posteriori LLRs,''
in \emph{Proc. 2025 IEEE Information Theory Workshop (ITW)},  2025.
\bibitem{b12}	P. Trifonov, “Efficient design and decoding of polar codes,” \emph{IEEE Trans.Commun.}, vol. 60, no. 11, pp. 3221–3227, Nov. 2012.
\bibitem{b13}   Y. S. Han, C. R. P. Hartmann, and C.-C. Chen, “Efficient priority-first search maximum-likelihood soft-decision decoding of linear block codes,” \emph{IEEE Trans. Inf. Theory}, vol. 39, no. 5, pp. 1514–1523, Sep. 1993.
    
    
    
\end{thebibliography}
\end{document}